\documentclass[aps,preprint,preprintnumbers,amsmath,amssymb,floatfix]{revtex4}
\usepackage[german,english]{babel}
\usepackage{bm,amsmath,amsfonts,dsfont,amsthm}
\usepackage{amssymb}
\usepackage{graphicx}
\usepackage{subfigure}
\usepackage{xcolor}

\DeclareMathOperator{\tr}{Tr}

\newtheorem{theorem}{Theorem}

\newtheorem{lemma}{Lemma}

\newtheorem{remark}{Remark}

\begin{document}

\title{Note on a Family of Monotone Quantum Relative Entropies} 

\author{Andreas~Deuchert, Christian~Hainzl, Robert~Seiringer}

\noaffiliation

\begin{abstract}
Given a convex function $\varphi$ and two hermitian matrices $A$ and $B$, Lewin and Sabin study in \cite{Mathieu_Julien} the relative entropy defined by $\mathcal{H}(A,B)=\tr \left[ \varphi(A) - \varphi(B) - \varphi'(B)(A-B) \right]$. Amongst other things, they prove that the so-defined quantity is monotone if and only if $\varphi'$ is operator monotone. The monotonicity is then used to properly define $\mathcal{H}(A,B)$ for bounded self-adjoint operators acting on an infinite-dimensional Hilbert space by a limiting procedure. More precisely, for an increasing sequence of finite-dimensional projections $\left\lbrace P_n \right\rbrace_{n=1}^{\infty}$ with $P_n \to 1$ strongly, the limit $\lim_{n \to \infty} \mathcal{H}(P_n A P_n, P_n B P_n)$ is shown to exist and to be independent of the sequence of projections $\left\lbrace P_n \right\rbrace_{n=1}^{\infty}$. The question whether this sequence converges to its "obvious" limit, namely $\tr \left[ \varphi(A)- \varphi(B) - \varphi'(B)(A-B) \right]$, has been left open. We answer this question in principle affirmatively and show that $\lim_{n \to \infty} \mathcal{H}(P_n A P_n, P_n B P_n) = \tr\left[ \varphi(A) - \varphi(B) - \frac{\text{d}}{\text{d} \alpha} \varphi\left( \alpha A + (1-\alpha)B \right)\vert_{\alpha = 0} \right]$. If the operators $A$ and $B$ are regular enough, that is $(A-B)$, $\varphi(A)-\varphi(B)$ and $\varphi'(B)(A-B)$ are trace-class, the identity $\tr\left[ \varphi(A) - \varphi(B) - \frac{\text{d}}{\text{d} \alpha} \varphi\left( \alpha A + (1-\alpha)B \right)\vert_{\alpha = 0} \right] = \tr \left[ \varphi(A)- \varphi(B) - \varphi'(B)(A-B) \right]$ holds. 
\end{abstract}

\maketitle

\section{Introduction and Main Results}
\label{sec:1}
We start with a quick review of the setting and the results of \cite{Mathieu_Julien} that are of interest for us. Let $\varphi \in C^0\left( [0,1], \mathbb{R} \right)$ be a continuous, convex function such that $\varphi'$ is continously differentiable on $(0,1)$ and let $A$ and $B$ be two hermitian matrices with $0 \leq A,B \leq 1$. Lewin and Sabin define a family of relative entropies of $A$ with respect to $B$ by the formula
\begin{equation}
\mathcal{H}(A,B) = \tr \left[ \varphi(A) - \varphi(B) - \varphi'(B)(A-B) \right].
\label{eq:1}
\end{equation}
As long as $0 < B < 1$, the above expression is well defined. If $0$ and/or $1$ are contained in the spectrum of $B$ and if $\varphi$ is not differentiable at these points this is still true if $A=B$ on $\text{Ker}(B)$, $\text{Ker}(1-B)$ or $\text{Ker}(B) \oplus \text{Ker}(1-B)$, respectively (the trace is taken on the complement of these subspaces). Are the just mentioned conditions not fulfilled, they define $\mathcal{H}(A,B) = \infty$.

In \cite[Theorem 1]{Mathieu_Julien}, the authors show that the so-defined relative entropy is monotone if and only if $\varphi'$ is operator monotone. We quote:
\begin{theorem}(Monotonicity).
Under the above conditions, the following are equivalent
\begin{enumerate}
\item $\varphi'$ is operator monotone on $(0,1)$;
\item For any linear map $X:h_1 \to h_2$ on finite-dimensional spaces $h_1$ and $h_2$ with $X^* X \leq 1$, and for any $0 \leq A,B \leq 1$ on $h_1$, we have
\begin{equation}
\mathcal{H}(X A X^* , X B X^*) \leq \mathcal{H}(A,B),
\label{eq:2}
\end{equation}
\end{enumerate}
with $\mathcal{H}(A,B)$ defined in Eq.~\eqref{eq:1}.
\label{thm:1}
\end{theorem}
In a second step, this result is used to extend the definition of the relative entropy to self-adjoint operators acting on an infinite-dimensional separable Hilbert space $h$ via the formula 
\begin{equation}
\mathcal{H}(A,B) := \lim_{n \to \infty} \mathcal{H}(P_n A P_n, P_n B P_n),
\label{eq:3}
\end{equation}
where $\left\lbrace P_n \right\rbrace_{n=1}^{\infty}$ is an increasing sequence of finite-dimensional projections with $P_n \to 1$ in the strong operator topology. By $\mathcal{L}(h)$ we denote the set of bounded linear operators on $h$ and $h_1, h_2$ denote infinite-dimensional separable Hilbert spaces. We quote again:
\begin{theorem}(Generalized relative entropy in infinite dimension).
We assume that $\varphi \in C^0\left( [0,1], \mathbb{R} \right)$ and that $\varphi'$ is operator monotone on $(0,1)$.
\begin{enumerate}
\item ($\mathcal{H}$ is well defined). For an increasing sequence $P_n$ of finite-dimensional projections on $h$ such that $P_n \to 1$ strongly, the sequence $\mathcal{H}(P_n A P_n, P_n B P_n)$ is monotone and possesses a limit in $\mathbb{R}^+ \cup \left\lbrace +\infty \right\rbrace$. This limit does not depend on the chosen sequence $P_n$ and hence $\mathcal{H}(A,B)$ is well-defined in $\mathbb{R}^+ \cup \left\lbrace +\infty \right\rbrace$.
\item (Approximation). If $X_n:h_1 \to h_2$ is a sequence such that $X_n ^* X_n \leq 1$ and $X_n^* X_n \to 1$ strongly in $h_1$, then
\begin{equation}
\mathcal{H}(A,B) = \lim_{n \to \infty} \mathcal{H}(X_n A X_n, X_n B X_n).
\label{eq:4} 
\end{equation}
\item (Weak lower semi-continuity). The relative entropy is weakly lower semi-continuous: if $0 \leq A_n, B_n \leq 1$ are two sequences such that $A_n \rightharpoonup A$ and $B_n \rightharpoonup B$ weakly-$*$ in $\mathcal{L}(h)$, then
\begin{equation}
\mathcal{H}(A,B) \leq \liminf_{n \to \infty} \mathcal{H}(A_n, B_n).
\label{eq:5}
\end{equation}
\end{enumerate}
\label{thm:2}
\end{theorem}
As one would expect, $\mathcal{H}(A,B)$ can take finite values when $A$ and $B$ themselves are not compact, as the following upper bound shows \cite[Theorem~3]{Mathieu_Julien}:
\begin{equation}
\mathcal{H}(A,B) \leq C \tr \left(\frac{1}{B^2} + \frac{1}{(1-B)^2} \right) \left( A-B \right)^2.
\label{eq:6}
\end{equation} 
Note that the only dependence on $\varphi$ on the right hand side of Eq.~\eqref{eq:6} is in the constant $C$. The question whether their notion of relative entropy in infinite dimensions is related to $\tr \left[ \varphi(A) - \varphi(B) - \varphi'(B)(A-B) \right]$, which is a-priori well-defined when the operator under the trace is trace-class, has been left open by the authors.

We answer this question in principle affirmatively, where "in principle" stands for the fact that $\tr \left[ \varphi(A) - \varphi(B) - \varphi'(B)(A-B) \right]$ turns out not to be the correct limit, in general.
\begin{theorem}
Let $\varphi \in C^0([0,1],\mathbb{R})$ be such that $\varphi'$ is operator monotone on (0,1) and let $\lbrace P_n \rbrace_{n=1}^{\infty}$ be defined as in Theorem~\ref{thm:2}. Then
\begin{equation}
\lim_{n \to \infty} \mathcal{H}(P_n A P_n, P_n B P_n) = \tr \left[ \varphi(A) - \varphi(B) - \frac{\text{d}}{\text{d} \alpha} \varphi\left( \alpha A + (1-\alpha)B \right)\Big|_{\alpha=0} \right],
\label{eq:7}
\end{equation}
with the understanding that either both sides are finite and equal each other, or both sides are infinite.
\label{thm:3}
\end{theorem}
\begin{remark}
We define the differential in Eq.~\eqref{eq:7} by the formula $\frac{\text{d}}{\text{d} \alpha} \left( \psi, \varphi\left( \alpha A + (1-\alpha)B \right) \psi \right) \big|_{\alpha=0} = \lim_{\alpha \to 0} \alpha^{-1} \left(\psi, \left[ \varphi(\alpha A + (1-\alpha)B) - \varphi(B) \right] \psi \right)$. In case $\varphi'$ is continuous on $[0,1]$, this limit exists for all $\psi \in h$. If $\varphi'$ is not continuous on the whole interval, it has singularities at $0$ and/or $1$ (we remind that $\varphi'$ is monotone increasing and continuous on $(0,1)$ by assumption) and we have to distinguish between three cases. First, assume $B$ has no eigenvalues at the points of discontinuity of $\varphi'$. Then the above limit exists for all $\psi$ in a suitably chosen dense set $D \subset h$ (see Section~\ref{sec:2}, Lemma~\ref{lem:3} for more details). Second, if $B$ has an eigenvalue at a point of discontinuity of $\varphi'$, $\psi$ is the corresponding eigenvector to the just mentioned eigenvalue and $(A-B)\psi \neq 0$, then the above limit equals $-\infty$. Third, if $(A-B)\psi = 0$ in the just mentioned situation, the above limit is equal to zero.
\label{rem:1}
\end{remark}
\begin{remark}
The operator monotonicity of the function $\varphi'$ implies the operator convexity of its primitives which in turn implies that $\varphi(A) - \varphi(B) - \frac{\text{d}}{\text{d} \alpha} \varphi\left( \alpha A + (1-\alpha)B \right)\vert_{\alpha=0}$ is positive (see Section~\ref{sec:2} for more details). This property can now be used to define a notion of trace that is applicable on the right hand side of Eq.~\eqref{eq:7}. Assume for the moment that $B$ has no eigenvalues at points of discontinuity of $\varphi'$ with $(A-B) \neq 0$ on the corresponding eigenspaces. Then the symmetric operator $\varphi(A) - \varphi(B) - \frac{\text{d}}{\text{d} \alpha} \varphi\left( \alpha A + (1-\alpha)B \right)|_{\alpha=0}$ can be defined on the dense set $D$ mentioned in Remark~\ref{rem:1} and since it is positive it has a Friedrichs extension $\left( T, \mathcal{D}(T) \right)$. By restricting attention to bases $\left\lbrace e_{\beta} \right\rbrace_{\beta = 1}^{\infty}$ where all $e_{\beta}$ lie in the form domain of $T$, we can define the trace of the operator on the right hand side of Eq.~\eqref{eq:7} to be $\sum_{\beta = 1}^{\infty} \left( e_{\beta}, T e_{\beta} \right)$, see Section~\ref{sec:2} for more details. The so-defined trace equals the usual trace whenever $\varphi(A) - \varphi(B) - \frac{\text{d}}{\text{d} \alpha} \varphi\left( \alpha A + (1-\alpha)B \right)|_{\alpha=0}$ is trace class and $+\infty$ otherwise. Now if $B$ has an eigenvalue at a point of discontinuity of $\varphi'$ and $(A-B) \neq 0$ on the corresponding eigenspace Remark~\ref{rem:1} suggest to define the trace on the right hand side of Eq.~\eqref{eq:7} to be $+\infty$. This goes hand in hand with the definition of the relative entropy for hermitian matrices of Lewin and Sabin which has been explained in the beginning of the introduction.
\label{rem:2}
\end{remark}
\begin{remark}
The idea to define the relative entropy as a trace over a manifestly positive operator in order to make it well-defined on a larger set, has already been used in \cite{HainzlLewinSeiringer}. The formula in the just mentioned reference equals the trace over $\varphi(A) - \varphi(B) - \frac{\text{d}}{\text{d} \alpha} \varphi\left( \alpha A + (1-\alpha)B \right)\vert_{\alpha=0}$ with $\varphi(x) = x \ln(x) + (1-x)\ln(1-x)$ and resembles Eq.~\eqref{eq:7b}.
\label{rem:3}
\end{remark}
\begin{remark} Assuming matrices, the equality $\tr \left[ \varphi(A) - \varphi(B) - \varphi'(B)(A-B) \right] = \tr \left[ \varphi(A) - \varphi(B) - \frac{\text{d}}{\text{d} \alpha} \varphi\left( \alpha A + (1-\alpha)B \right)\vert_{\alpha=0} \right]$ holds as one can see with a direct computation that exploits the cyclicity of the trace (see \cite[Theorem V.3.3]{Bhatia1997} for a simple way to compute the derivative). In general however, this cannot be expected. 
\label{rem:4}
\end{remark}
A crucial ingredient of our proof of Theorem~\ref{thm:3} is the following Lemma which we state here because it is of interest in itself.
\begin{lemma}
Let $\varphi \in C\left([0,1],\mathbb{R}\right)$ be such that $\varphi'$ is operator monotone on $(0,1)$. Assume further that $B$ has no eigenvalues at points of discontinuity of $\varphi'$ with $(A-B) \neq 0$ on the corresponding eigenspaces. Then there exists a constant $b \geq 0$ and a unique Borel probability measure $\mu$ on $[-1,1]$ such that
\begin{align}
&\tr \left[ \varphi(A) - \varphi(B) - \frac{\text{d}}{\text{d}\alpha} \varphi(\alpha A + (1-\alpha)B) \Big|_{\alpha=0} \right] = \label{eq:7b} \\
&2 b \int_{-1}^{1} \int_{0}^{\infty} \tr \left[ \frac{1}{1+\lambda(1-2 B) + t} Q \frac{1}{1+\lambda(1-2 A) + t} Q \frac{1}{1+\lambda(1-2 B) + t} \right] \text{d}t \ \text{d}\mu(\lambda), \nonumber
\end{align}
where $Q=(A-B)$. (To be precise, $\mu$ is unique only if $b>0$.)
\label{lem:1}
\end{lemma}
\begin{remark}
The formula on the right hand side of Eq.~\eqref{eq:7b} is in many circumstances easier to handle than the formula on the left hand side. This is because it is the integral (with a positive measure) of a positive function which is the trace of a bounded positive operator. In particular, the operator under the trace has a simpler form than the one on the left hand side of Eq.~\eqref{eq:7b}. 
\label{rem:4b}
\end{remark}
Assuming more regular operators $A$ and $B$, the equality mentioned in Remark~\ref{rem:4} is still true in infinite dimensions as the following statement shows.
\begin{theorem}
Let $\varphi \in C^0([0,1],\mathbb{R})$ be such that $\varphi'$ is operator monotone on (0,1). Assume in addition that $(A-B)$, $\varphi(A)-\varphi(B)$ and $\varphi'(B)(A-B)$ are trace-class. Then $ \frac{\text{d}}{\text{d} \alpha} \varphi\left( \alpha A + (1-\alpha)B \right)|_{\alpha=0}$ is trace-class and the identity
\begin{equation}
\tr \left[ \varphi(A) - \varphi(B) - \frac{\text{d}}{\text{d} \alpha} \varphi\left( \alpha A + (1-\alpha)B \right)\Big|_{\alpha=0} \right] = \tr \left[ \varphi(A)- \varphi(B) - \varphi'(B)(A-B) \right]
\label{eq:8}
\end{equation}
\label{thm:4}
holds.
\end{theorem}
\begin{remark}
In mathematical physics one encounters applications where the state of a physical system is defined to be a minimizer of a nonlinear functional in which the physical relative entropy appears, see e.g. \cite{HainzlLewinSeiringer,GLderivation}. For a fermionic many-particle system the function $\varphi(x)=x \ln(x) + (1-x) \ln(1-x)$ is the right choice to define the physical relative entropy, while for bosons it is $\varphi(x)=x \ln(x) - (1+x) \ln(1+x)$. We note that both functions fulfill the requirements of Theorems~\ref{thm:2}-\ref{thm:4}. Since the right hand side of Eq.~\eqref{eq:8} is in practice much easier to evaluate explicitly than the left hand side when given a trial state $A$, Theorem~\ref{thm:4} becomes important if one wants to derive an upper bound for the minimal energy of such a functional. The left hand side of Eq.~\eqref{eq:7} in contrast is important since it allows one to prove upper or lower bounds for the relative entropy with the help of Klein's inequality, see \cite{Mathieu_Julien}.     
\label{rem:5}
\end{remark}
\section{Proof of Theorem~\ref{thm:3}}
\label{sec:2}
The main ingredient of our proof is the derivation of the formula stated in Lemma~\ref{lem:1}. Having this identity at hand, we show the convergence of the relative entropy by first showing it for the trace under the integral. In a second step, we argue why the limit can be interchanged with the integrals over $\lambda$ and $t$. At this point Theorem~\ref{thm:1} enters the analysis in a crucial way.

Since $\varphi'$ is operator monotone on $(0,1)$ there exists a unique Borel probability measure $\mu$ on $[-1,1]$ such that (see \cite[Corollary V.4.5]{Bhatia1997}) 
\begin{equation}
\varphi'(x) = a + b \int_{-1}^{1} \frac{2x-1}{1- \lambda(2x-1)} \text{d}\mu(\lambda), 
\label{eq:10}
\end{equation}
with $b \geq 0$ (To be precise, $\mu$ is unique only if $b>0$.). When integrating the above expression, one obtains a primitive for $\varphi'$ which is of the form
\begin{equation}
\varphi(x) = ax+c-\frac{b}{2} \int_{-1}^{1} \left( \frac{2x-1}{\lambda} + \frac{\ln \left(1+\lambda(1-2x) \right)}{\lambda^2} \right) \text{d}\mu(\lambda).
\label{eq:11}
\end{equation}
Since $x \mapsto - \ln(x)$ is an operator convex function, the same holds true for $\varphi$. 

To keep the main argumentation straight, we first prove two technical Lemmata. The first concerns the relation between the regularity of $\varphi$ at the endpoints of the interval $[0,1]$ and the behavior of the measure $\mu$ in the vicinity of $-1$ and $1$.
\begin{lemma}
Assume $\varphi \in C^0([0,1],\mathbb{R})$ such that $\varphi'$ is operator monotone on $(0,1)$. Then $\mu(\lbrace -1 \rbrace) = 0 = \mu(\lbrace 1 \rbrace)$,
\begin{equation}
\int_{1/2}^1 -\ln(1-\lambda) \text{d}\mu(\lambda) < \infty \quad \quad \text{and} \quad \quad \int_{-1}^{-1/2} -\ln(1+\lambda) \text{d}\mu(\lambda) < \infty.
\label{eq:12}
\end{equation}
If in addition $\varphi' \in C^0([0,1],\mathbb{R})$, the stronger implications
\begin{equation}
\int_{1/2}^1 \frac{1}{1-\lambda} \text{d}\mu(\lambda) < \infty \quad \quad \text{and} \quad \quad \int_{-1}^{-1/2} \frac{1}{1+\lambda} \text{d}\mu(\lambda) < \infty
\label{eq:13}
\end{equation}
hold. In case $\varphi'$ is not continuous at $1$ the first integral in Eq.~\eqref{eq:13} equals $+\infty$ and if it is not continuous at $0$ this is true for the second integral.   
\label{lem:2}
\end{lemma}
\begin{proof} 
We start with the first case, hence we assume that only $\varphi$ is continuous on $[0,1]$. Since the limits $\lim_{x \to 0} \varphi(x)$ and $\lim_{x \to 1} \varphi(x)$ exist we can conclude that $\mu(\lbrace -1 \rbrace) = 0 = \mu(\lbrace 1 \rbrace)$ holds. We further conclude that the following limit exists (see Eq.~\eqref{eq:11})
\begin{equation}
\infty > \lim_{x \to 1} \int_{1/2}^1 -\frac{\ln(1+\lambda(1-2x))}{\lambda^2} \text{d}\mu(\lambda) \geq \frac{1}{4} \int_{1/2}^1 -\ln(1-\lambda) \text{d}\mu(\lambda).
\label{eq:14}
\end{equation}
To come to the expression on the right hand side, we have applied Fatou's Lemma. Doing the same argumentation again, this time with the limit $x \to 0$, yields $\int_{-1}^{-1/2} \ln(1+ \lambda) \text{d}\mu(\lambda)< \infty$. If also $\varphi'$ is continuous on $[0,1]$, we compute
\begin{equation}
\lim_{x \to 1} \varphi'(x) = a+b \lim_{x \to 1} \int_{-1}^{1} \frac{2x-1}{1-\lambda (2x-1)} \text{d}\mu(\lambda) \geq a + \int_{-1}^1 \frac{1}{1-\lambda} \text{d}\mu(\lambda),
\label{eq:15}
\end{equation}
where, as before, we have applied Fatou's Lemma. The same procedure with $\lim_{x \to 0} -\varphi'(x)$ yields the other bound. Using the monotonicity of the integrand, one easily shows that the just discussed integrals diverge to $+\infty$ in case $\varphi'$ is not continuous at $0$ and/or $1$, respectively.  
\end{proof}

In order to obtain a handy formula for the operator $\frac{\text{d}}{\text{d} \alpha} \varphi\left( \alpha A + (1-\alpha)B \right)\vert_{\alpha = 0}$, we explicitly compute the directional derivative. 
\begin{lemma}
Assume $\varphi \in C^0([0,1],\mathbb{R})$ such that $\varphi'$ is operator monotone on $(0,1)$. If $\varphi'$ is continuous on $[0,1]$ or if it is discontinuous at $0$ and/or $1$ and $B$ has no eigenvalue at these points then
\begin{align}
&\frac{\text{d}}{\text{d}\alpha} \left( \psi, \varphi(\alpha A + (1-\alpha)B) \psi \right) \Big|_{\alpha=0} = a \left( \psi, (A-B) \psi \right) -\frac{b}{2} \int_{-1}^{1} \Bigg[ \left( \psi, \frac{2(A-B)}{\lambda} \psi \right) \label{eq:16} \\
&\hspace{3cm} -\frac{2}{\lambda} \int_{0}^{\infty} \left( \psi, \frac{1}{1+\lambda(1-2B) + t}(A-B) \frac{1}{1+\lambda(1-2B) + t} \psi \right) \text{d}t \Bigg] \text{d}\mu(\lambda), \nonumber
\end{align}
where $a$,$b$ and $\mu$ are defined by Eq.~\eqref{eq:10}. In case of the first scenario ($\varphi'$ continuous on $[0,1]$), the derivative is taken for all $\psi \in h$ while in the second scenario it is taken only for all $\psi$ in a dense set $D \subset h$. Explicitly, the set $D$ is given by $D = \cup_{\epsilon > 0} \mathds{1}(\epsilon < B < 1-\epsilon) h$ in case $0$ and $1$ are points of discontinuity of $\varphi'$ and by the obvious generalization when $\varphi'$ is discontinuous only at one of these points. This accounts for the fact that the limiting operator may be unbounded. In case $\varphi'$ has discontinuities and $B$ has eigenvalues at at least one of these points, we have to treat the above limit with $\psi$ being one of the eigenvectors to the just mentioned eigenvalues separately. We distinguish between two cases. If $(A-B) \psi \neq 0$ we have
\begin{equation}
-\lim_{\alpha \to 0} \left( \psi, \left[ \frac{\varphi(\alpha A + (1-\alpha B)) - \varphi(B)}{\alpha} \right] \psi \right) = \infty.
\label{eq:17}
\end{equation} 
If $(A-B) \psi =0$ instead, the limit in Eq.~\eqref{eq:17} equals zero.
\label{lem:3}
\end{lemma}
\begin{proof}
Using Eq.~\eqref{eq:11}, one can easily check the identity
\begin{align}
&\frac{\text{d}}{\text{d}\alpha} \left( \psi, \varphi\left( \alpha A + (1-\alpha)B \right) \psi \right) \Big|_{\alpha=0} = a \left( \psi, (A-B) \psi \right) - \frac{b}{2} \lim_{\alpha \to 0} \int_{-1}^{1} \Bigg[ \left( \psi, \frac{2(A-B)}{\lambda} \psi \right) \label{eq:18} \\
&\hspace{3cm} + \left( \psi, \frac{\ln\left( 1+\lambda (1-2(B+\alpha(A-B))) \right) - \ln\left( 1+\lambda(1-2B) \right)}{\alpha \lambda^2} \psi \right) \Bigg] \text{d}\mu(\lambda). \nonumber 
\end{align}
The second term is just the difference quotient defining the directional derivative of the second term in Eq.~\eqref{eq:11}. Let us have a closer look at the term with the logarithms. We use the formula $\ln(x)=\int_0^{\infty} \left( \frac{1}{1+t} - \frac{1}{x+t} \right) \text{d}t$ and apply the resolvent identity once, to see that it can be written as
\begin{align}
&\left( \psi, \frac{\ln\left( 1+\lambda (1-2(B+\alpha(A-B))) \right) - \ln\left( 1+\lambda(1-2B) \right)}{\alpha \lambda^2} \psi \right) = \label{eq:19} \\ 
&\hspace{1.5cm}-\frac{2}{\lambda} \int_{0}^{\infty} \left( \psi, \frac{1}{1+\lambda(1-2B) + t}(A-B) \frac{1}{1+\lambda(1-2(B+\alpha(A-B))) + t} \psi \right) \text{d}t. \nonumber
\end{align}
In order to explicitly compute the limit $\alpha \to 0$, it needs to be interchanged in a first step with the integral over $\lambda$ and in a second step with the integral over $t$. The second step will follow easily from the estimates used to show the first step since $\lambda \in (-1,1)$ is then fixed which implies that all resolvents are uniformly bounded. We therefore focus on the interchange of the limit $\alpha \to 0$ with the integral over $\lambda$. In order to be able to apply dominated convergence, we have to find a positive function $g \in L^1(\mu)$ with 
\begin{align}
\left| \left( \psi, \frac{2(A-B)}{\lambda} \psi \right) -\frac{2}{\lambda} \int_0^{\infty} \left( \psi, R(B)(A-B)R(B + \alpha(A-B)) \psi\right) \text{d}t \right| \leq g(\lambda) 
\label{eq:20}
\end{align}  
for all $\psi$ at least in a dense subset of $h$ (The case where $B$ has eigenvalues at points of discontinuity of $\varphi'$ will be treated at the end.). To shorten the writing, we have introduced the notation $R(B) = (1+\lambda(1-2B) + t)^{-1}$.

Let us first investigate the behavior of our integrand for $\lambda \in \left(-1+\epsilon,1-\epsilon \right)$. We write $R(B) = \frac{1}{1+t} - \frac{\lambda}{1+t} (1-2B) R(B)$ (and the same for $R(B+\alpha(A-B))$) and evaluate the contribution of the first term which reads
\begin{equation}
-\frac{2}{\lambda} \int_0^{\infty} \frac{1}{1+t} \left( \psi, (A-B)\psi \right) \frac{1}{1+t} \text{d}t = -\frac{2}{\lambda} \left( \psi,(A-B) \psi\right).
\label{eq:21}
\end{equation}
It cancels the first term under the integral on the right hand side of Eq.~\eqref{eq:18}. The three remaining terms have no singularity and can be bounded by a constant. 

In the vicinity of $\lambda = -1$ and $\lambda = 1$ the situation is a little different and one needs to argue more carefully. We will distinguish three cases depending on the regularity of $\varphi'$ at $0$ and $1$ and on the spectrum of $B$. First let us assume that $\varphi'$ is not continuous at $0$ and $1$ and that $B$ has no eigenvalues at these points. Let $D_{\epsilon} = \mathds{1}\left( \epsilon < B < 1-\epsilon \right) h$ and define $D= \cup_{\epsilon > 0} D_{\epsilon}$. Due to our assumptions on $B$, the set $D$ is dense in $h$. For $\psi \in D$, we investigate 
\begin{align}
\int_{0}^{\infty}& \left| \left( \psi, \frac{1}{1+\lambda (1-2B)+t} (A-B) \frac{1}{1+\lambda (1-2(\alpha A + (1-\alpha)B))+t} \psi \right) \right| \text{d}t \label{eq:22} \\
&\leq \int_{0}^{\infty} \left\Vert \frac{1}{1+\lambda (1-2B)+t} \psi \right\Vert \left\Vert A-B \right\Vert_{\infty} \left\Vert \frac{1}{1+\lambda (1-2(\alpha A + (1-\alpha)B))+t} \psi \right\Vert \text{d}t   
\nonumber
\end{align}
which is the relevant contribution from Eq.~\eqref{eq:20}. The part of the integral over $t$ from say $1$ to $\infty$ is easy to control. One just bounds the resolvents in operator norm by $1/t$. After the evaluation of the integral, we end up with a constant. To bound the other part of the integral over $t$ (the one from $0$ to $1$), we use the fact that $\psi \in D$ which implies that $\left\Vert \frac{1}{1+\lambda (1-2B)+t} \psi \right\Vert \leq 1/\epsilon$ for an $\epsilon > 0$ that depends on $\psi$. On the other hand $\left\Vert \frac{1}{1+\lambda (1-2(\alpha A + (1-\alpha)B))+t} \psi \right\Vert \leq \frac{1}{1+\lambda + t}$ for $\lambda$ close to $-1$. Putting this together, we obtain
\begin{align}
\int_{0}^{\infty} \Bigg\Vert \frac{1}{1+\lambda (1-2B)+t} \psi \Bigg\Vert \Big\Vert A-B \Big\Vert_{\infty} &\Bigg\Vert \frac{1}{1+\lambda (1-2(\alpha A + (1-\alpha)B))+t} \psi \Bigg\Vert \text{d}t \label{eq:23} \\
&\leq  \left\Vert A-B \right\Vert_{\infty} \left( \frac{1}{\epsilon} \int_{0}^{1} \frac{1}{1+\lambda+t} \text{d}t + C \right) \nonumber \\ 
&\leq C(\epsilon) \left( -\ln(1+\lambda) + 1 \right). \nonumber
\end{align}
A similar bound can be obtained for $\lambda$ close to $1$. There the function $-\ln(1-\lambda)$ enters the analysis. Hence, there exists a constant $C(\epsilon)$ depending on $\psi$ such that
\begin{align}
\Bigg| \Bigg( \psi, \frac{2(A-B)}{\lambda} \psi \Bigg) -\frac{2}{\lambda} \int_0^{\infty} \Bigg( \psi, R(B)(A-B)R(B + &\alpha(A-B)) \psi\Bigg) \text{d}t \Bigg| \label{eq:24} \\
&\leq C(\epsilon) \left( -\ln(1-|\lambda|) + 1 \right). \nonumber
\end{align}
Because of Lemma~\ref{lem:2}, the bound allows us to take the limit inside the integral and proves the claim in this situation.

Nearly the same argumentation goes through when $B$ has spectrum at $0$ and/or $1$ and if $\varphi'$ is continuous at these points. By bounding both resolvents like we did with the second in the previous step, that is $\left\Vert R(B)  \psi \right\Vert \leq (1-|\lambda| +t)^{-1}$ and the same with $\left\Vert R(\alpha A + (1-\alpha)B) \psi \right\Vert$, one obtains 
\begin{equation}
\Bigg| \Bigg( \psi, \frac{2(A-B)}{\lambda} \psi \Bigg) -\frac{2}{\lambda} \int_0^{\infty} \Bigg( \psi, R(B)(A-B)R(B + \alpha(A-B)) \psi\Bigg) \text{d}t \Bigg| \leq\frac{C}{1-|\lambda|}. \label{eq:25}
\end{equation}
Again due to Lemma~\ref{lem:2}, this is enough to interchange the limit and the integral. The case where $\varphi'$ is discontinuous only at one point is treated in the obvious way.

For the last case we have to assume that $\varphi'$ is not continuous at $0$ and/or $1$ and that $B$ has an eigenvalue at at least one of these points. We only investigate the relevant contribution. Let $\psi$ be the eigenvector of $B$ to the eigenvalue $0$ for example (the other cases go the same way). We will show that
\begin{equation}
\lim_{\alpha \to 0} \int_{-1}^{-1/2} \frac{-1}{\lambda} \int_0^{\infty} \left( \psi, R(B) (A-B) R(B+\alpha (A-B)) \psi \right) \text{d}t \ \text{d}\mu(\lambda) = \infty,
\label{eq:26}
\end{equation}
if $(A-B) \psi \neq 0$ and that the above limit equals zero in case $(A-B) \psi =0$. Using $R(B + \alpha(A-B)) = R(B) + 2 \alpha \lambda R(B + \alpha(A-B)) (A-B) R(B)$, the integrand can be written as
\begin{align}
&\frac{-1}{\lambda} \left( \psi, R(B)(A-B)R(B) \psi \right) - 2 \alpha \left( \psi, R(B) (A-B) R(B + \alpha(A-B)) (A-B) R(B) \psi \right) \nonumber \\
&\hspace{3cm} =\left( \frac{1}{1+\lambda+t} \right)^2 \left[ \frac{-1}{\lambda} \left( \psi, A \psi \right) - 2 \alpha \left( \psi, A R(B + \alpha(A-B)) A \psi \right) \right]. \label{eq:27}
\end{align}
Let us first assume that $(A-B) \psi \neq 0$ which implies that $\left( \psi, A \psi \right) >0$. Since the function $t \mapsto \frac{1}{t}$ is operator convex on the interval $(0,\infty)$, see \cite[Exercise~V.2.11]{Bhatia1997}, we know that $R(\alpha A + (1-\alpha)B) \leq \alpha R(A) + (1-\alpha) R(B)$. If we apply this inequality on the right hand side of Eq.~\eqref{eq:27} and discard all positive terms in order to obtain a lower bound, we find 
\begin{align}
\bigg( &\frac{1}{1+\lambda+t} \bigg)^2 \left[ \frac{-1}{\lambda} \left( \psi, A \psi \right) - 2 \alpha \left( \psi, A R(B + \alpha(A-B)) A \psi \right) \right] \label{eq:27b} \\
&\hspace{1cm}\geq \left( \frac{1}{1+\lambda+t} \right)^2 \left[ \frac{-1}{\lambda} \left( \psi, A \psi \right) - 2 \alpha^2 \left( \psi, A R(A) A \psi \right) - 2 \alpha \left( \psi, A R(B) A \psi \right) \right]. \nonumber
\end{align}
The right hand side of this equation, viewed as a function of $\alpha$, is certainly monotone and so we can use monotone convergence to show that
\begin{align}
\lim_{\alpha \to 0} \int_{-1}^{-1/2} \int_0^{\infty} \left( \frac{1}{1+\lambda+t} \right)^2 &\bigg[ \frac{-1}{\lambda} \left( \psi, A \psi \right) - 2 \alpha^2 \left( \psi, A R(A) A \psi \right) \label{eq:28} \\
&\hspace{2.3cm}- 2 \alpha \left( \psi, A R(B) A \psi \right) \bigg] \text{d}t \ \text{d}\mu(\lambda) \nonumber \\
&= \left(\psi, A \psi \right) \int_{-1}^{-1/2} \frac{-1}{\lambda} \int_{0}^{\infty} \left( \frac{1}{1+\lambda + t} \right)^2 \text{d}t \ \text{d}\mu(\lambda) \nonumber \\
&\geq \left( \psi, A \psi \right) \int_{-1}^{-1/2} \frac{1}{1+\lambda} \text{d}\mu(\lambda) = \infty. \nonumber
\end{align}
The last equality is achieved with the help of Lemma~\ref{lem:2}. Now assume that $(A-B) \psi= 0$ which means that $A \psi = 0$. Hence, $\left[ \alpha A + (1-\alpha) B \right] \psi = 0$ and $\varphi(\alpha A + (1-\alpha) B) \psi = \varphi(0) \psi$. Since this expression is a constant the derivative with respect to $\alpha$ vanishes. A similar argument can be done when $B$ has $1$ as an eigenvalue. This concludes the proof of Lemma~\ref{lem:3}.
\end{proof}

Before we come to the main part of the proof, we have to argue how the trace on the right hand side of Eq.~\eqref{eq:7} can be defined. Let us for the moment assume that $B$ has no eigenvalues at points of discontinuity of $\varphi'$ with $(A-B) \neq 0$ on the corresponding eigenspaces. Then by Lemma~\ref{lem:3}, we can define the quadratic form
\begin{equation}
q(\psi,\eta) = \lim_{\alpha \to 0} \left( \psi, \left[ \varphi(A) - \varphi(B) -   \frac{\varphi(B+\alpha(A-B)) - \varphi(B)}{\alpha} \right] \eta \right) \label{eq:28b}
\end{equation}
on the dense set $D \subset h$ (The set $D$ has been defined in Lemma~\ref{lem:3}.). The operator convexity of $\varphi$ implies that $\frac{\varphi(B+\alpha(A-B)) - \varphi(B)}{\alpha} \leq \varphi(A) - \varphi(B)$ holds for all $0 < \alpha \leq 1$. Since the inequality is preserved by the limiting procedure $\alpha \to 0$ we conclude that $q$ is positive. It is an easy exercise to check with the methods used in the proof of Lemma~\ref{lem:3}, that on $D$, the operator 
\begin{align}
\varphi(A)& - \varphi(B) - a (A-B) \label{eq:28c} \\
&+ \frac{b}{2} \int_{-1}^{1} \Bigg[ \frac{2(A-B)}{\lambda}  -\frac{2}{\lambda} \int_{0}^{\infty} \frac{1}{1+\lambda(1-2B) + t}(A-B) \frac{1}{1+\lambda(1-2B) + t} \text{d}t \Bigg] \text{d}\mu(\lambda), \nonumber
\end{align}
is well-defined, symmetric and due to the previous reasoning also positive [compare with Eq.~\eqref{eq:16}]. Again by Lemma~\ref{lem:3}, its associated quadratic form is $q$. The theorem on the Friedrichs extension tells us that $q$ is closable and that its closure $(\hat{q},\mathcal{Q}(\hat{q}))$ is the quadratic form of a unique self-adjoint operator $(T,\mathcal{D}(T))$ whose domain $\mathcal{D}(T)$ is contained in the form domain $\mathcal{Q}(\hat{q})$ of $\hat{q}$, see \cite[Theorem~X.23]{ReedSimon2}. Additionally, $T$ is positive. Having the Friedrichs extension at hand, we can define the right hand side of Eq.~\eqref{eq:7} to be the trace of $T$. To that end, we restrict our attention to bases $\lbrace e_{\beta} \rbrace_{\beta = 1}^{\infty}$ of $h$ with $e_{\beta} \in \mathcal{Q}(\hat{q})$ for all $\beta \in \mathbb{N}$ and define $\text{Tr}(T)= \sum_{\beta = 1}^{\infty} \hat{q}(e_{\beta},e_{\beta})$. Of course, this definition does not depend on the choice of the basis. It yields the usual notion of trace when $T$ is trace-class and gives $\text{Tr}(T) = +\infty$ otherwise. If $B$ has an eigenvalue at a point of discontinuity of $\varphi'$ with $A-B \neq 0$ on the corresponding eigenspace, Lemma~\ref{lem:3} suggest to define the trace of the right hand side of Eq.~\eqref{eq:7} to be $+\infty$. This goes hand in hand with the definition of Lewin and Sabin mentioned in the beginning of the introduction. 

Having these prerequisites at hand, we come to the main part of our proof. If $B$ has an eigenvalue at a point of discontinuity of $\varphi'$ and $(A-B) \neq 0$ on the corresponding eigenspace then both sides of Eq.~\eqref{eq:7} equal $+\infty$. For the right hand side this has been discussed in the previous paragraph while for the left hand side, this can be seen by choosing $P_1$ such that the eigenspace of the just mentioned eigenvalue lies in its range. Hence, we can exclude this case in what follows. The key point of our proof is the more explicit formula Eq.~\eqref{eq:7b} for the trace of the operator $\varphi(A) - \varphi(B) - \frac{\text{d}}{\text{d}\alpha} \varphi(\alpha A + (1-\alpha) B) |_{\alpha=0}$ which we derive now. Using Eq.~\eqref{eq:11}, the operator $\varphi(A) - \varphi(B)$ can be written as
\begin{align}
\varphi(A) - \varphi(B) &= a(A-B) \label{eq:29} \\
&\ \ \ - \frac{b}{2} \int_{-1}^{1} \left[ \frac{2(A-B)}{\lambda} + \frac{\ln(1+\lambda(1-2A)) - \ln(1+\lambda(1-2B))}{\lambda^2} \right] \text{d}\mu(\lambda). \nonumber
\end{align}
In the next step, we write the difference of the two logarithms in Eq.~\eqref{eq:29} with the help of the formula $\ln(x) = \int_{0}^{\infty} \left( \frac{1}{1+t} - \frac{1}{x+t} \right) \text{d}t$ as an integral over resolvents. When we add the explicit representation for $\frac{\text{d}}{\text{d}\alpha} \varphi(\alpha A + (1-\alpha)B)|_{\alpha=0}$ that has been derived in Lemma~\ref{lem:3} and apply the resolvent identity twice, we arrive at the formula  
\begin{align}
&\varphi(A) - \varphi(B) - \frac{\text{d}}{\text{d}\alpha} \varphi(\alpha A + (1-\alpha)B)\big|_{\alpha=0} = \label{eq:30} \\
&\hspace{1cm} 2 b \int_{-1}^{1} \int_{0}^{\infty} \frac{1}{1+\lambda(1-2 B) + t} Q \frac{1}{1+\lambda(1-2 A) + t} Q \frac{1}{1+\lambda(1-2 B) + t} \text{d}t \ \text{d}\mu(\lambda), \nonumber
\end{align}
where we have introduced the shortcut $Q=(A-B)$. Taking the trace on both sides, we can commute it with the integrals because the integrand is a positive operator and obtain Eq.~\eqref{eq:7b}. Hence, we have proved Lemma~\ref{lem:1}.

Now let $\left\lbrace P_n \right\rbrace_{n=1}^{\infty}$ be an increasing sequence of finite-dimensional projections  that converges to $1$ in the strong operator topology. Because for matrices the two ways of writing the relative entropy are the same (see Remark~\ref{rem:4}) we have the formula
\begin{align}
\mathcal{H}(A_n,B_n) = 2 b \int_{-1}^{1} \int_{0}^{\infty} \tr \left[ R(B_n) Q_n R(A_n) Q_n R(B_n) \right] \text{d}t \ \text{d}\mu(\lambda)
\label{eq:31}
\end{align}
with $A_n = P_n A P_n$ and so on. We will first show that $\lim_{n \to \infty} \tr \left[ R(B_n) Q_n R(A_n) Q_n R(B_n) \right] = \tr \left[ R(B) Q R(A) Q R(B) \right]$ and then argue why we can interchange the limit with the two integrals. 

Let $m \geq 1$. In order to be able to restrict the trace on the right hand side of Eq.~\eqref{eq:31} to a finite-dimensional subspace, we first investigate
\begin{align}
\tr \big[ (1-P_m) R(B_n) Q_n &R(A_n) Q_n R(B_n) (1-P_m) \big] \label{eq:32} \\
&\leq \tr \left[ (1-P_m) R(B_n) Q_n^2 R(B_n) (1-P_m) \right] \frac{1}{1-|\lambda|+t} \nonumber \\
&\leq \tr \left[ (1-P_m) R(B_n) P_n Q^2 P_n R(B_n) (1-P_m) \right] \frac{1}{1-|\lambda|+t}. \nonumber
\end{align}
Let us for the moment assume that $Q$ is Hilbert-Schmidt which implies that it can be written as $Q=\sum_{\beta=1}^{\infty} q_{\beta} \vert \psi_{\beta} \rangle\langle \psi_{\beta} \vert$ with $\sum_{\beta = 1}^{\infty} q_{\beta}^2 < \infty$. The case when this does not hold true is taken care of at the end. Using the cyclicity of the trace, we write
\begin{align}
\tr \big[ (1-P_m) R(B_n) P_n Q^2 P_n R(B_n) (1-P_m) \big] &= \tr \left[ Q P_n R(B_n) (1-P_m) R(B_n) P_n Q \right] \label{eq:33} \\ 
&= \sum_{\alpha=1}^{k} \left( \psi_{\alpha}, Q P_n R(B_n) (1-P_m) R(B_n) P_n Q \psi_{\alpha} \right) \nonumber \\
&+ \sum_{\alpha=k+1}^{\infty} \left( \psi_{\alpha}, Q P_n R(B_n) (1-P_m) R(B_n) P_n Q \psi_{\alpha} \right). \nonumber
\end{align}
The term in the last line on the right hand side of Eq.~\eqref{eq:33} can be bounded uniformly in $n$ as the next calculation shows,
\begin{align}
\Bigg| \sum_{\alpha=k+1}^{\infty} \big( \psi_{\alpha}, Q P_n R(B_n) (1-P_m) &R(B_n) P_n Q \psi_{\alpha} \big) \Bigg| \label{eq:34} \\   
&= \left| \sum_{\alpha=k+1}^{\infty} q_{\alpha}^2 \left( \psi_{\alpha}, P_n R(B_n) (1-P_m) R(B_n) P_n \psi_{\alpha} \right) \right| \nonumber \\
&\leq \left( \frac{1}{1-|\lambda|+t} \right)^2 \sum_{\alpha=k+1}^{\infty} q_{\alpha}^2. \nonumber 
\end{align}
The right hand side of Eq.~\eqref{eq:34} goes to zero as $k$ tends to infinity for all $-1 < \lambda < 1$ and $t \geq 0$ due to the assumptions on $Q$. On the other hand, $\sum_{\alpha=1}^{k} \left( \psi_{\alpha}, Q P_n R(B_n) (1-P_m) R(B_n) P_n Q \psi_{\alpha} \right) \to \sum_{\alpha=1}^{k} \left( \psi_{\alpha}, Q P R(B) (1-P_m) R(B) Q \psi_{\alpha} \right)$ for $n \to \infty$ because the sum is finite and the operator in the middle is convergent in the strong operator topology, see \cite[Theorem VIII.20]{ReedSimon1}. When we consider Eq.~\eqref{eq:32} again and take the limit $n \to \infty$ followed by the limit $k \to \infty$, we arrive at
\begin{align}
\lim_{n \to \infty} \tr \big[ (1-P_m) R(B_n) &Q^2 R(B_n) (1-P_m) \big] \label{eq:35} \\
&= \tr \big[ (1-P_m) R(B) Q^2 R(B) (1-P_m) \big]. \nonumber
\end{align}
Let us denote the left hand side of this equation by $\delta(n,m)$ and the right hand side by $\delta(m)$. By construction, $\lim_{m \to \infty} \delta(m) = 0$ holds. Using this result, we easily get the following two inequalities
\begin{align}
&\tr \left[ R(B_n) Q_n R(A_n) Q_n R(B_n) \right] \leq \tr \left[P_m R(B_n) Q_n R(A_n) Q_n R(B_n) P_m \right] + \tilde{\delta}(n,m), \label{eq:36} \\
&\tr \left[ R(B_n) Q_n R(A_n) Q_n R(B_n) \right] \geq \tr \left[P_m R(B_n) Q_n R(A_n) Q_n R(B_n) P_m \right], \nonumber 
\end{align}
where $\tilde{\delta}(n,m) = \delta(n,m) (1-|\lambda|+t)^{-1}$. Taking first the limit $n \to \infty$ and then the limit $m \to \infty$ in the above equations, we conclude that $\lim_{n \to \infty} \tr \left[ R(B_n) Q_n R(A_n) Q_n R(B_n) \right] = \tr \left[ R(B) Q R(A) Q R(B) \right]$ for all $-1 < \lambda < 1$.

The next step in the proof is to interchange the limit $n \to \infty$ and the integrals. Let us start with the integral over $t$. Since we only need a bound for almost every $\lambda$ to apply dominated convergence we can assume that $-1<\lambda<1$. Under these conditions a dominating function is easily constructed because
\begin{align}
\tr  \left[ R(B_n) Q_n R(A_n) Q_n R(B_n) \right] \leq \left( \frac{1}{1-|\lambda|+t} \right)^3 \left\Vert Q \right\Vert_2^2.
\label{eq:37}
\end{align}
Hence, we have shown that
\begin{align}
\int_{-1}^1 \int_0^{\infty} \tr  &\left[ R(B) Q R(A) Q R(B) \right] \text{d}t \text{d}\mu(\lambda) \label{eq:38} \\
&\hspace{2.9cm} = \int_{-1}^1 \lim_{n \to \infty} \left( \int_0^{\infty} \tr  \left[ R(B_n) Q_n R(A_n) Q_n R(B_n) \right] \text{d}t \right) \text{d}\mu(\lambda). \nonumber
\end{align}

To interchange the limit with the first integral, we have to argue more carefully and use the monotonicity of the relative entropy. With similar but somewhat easier arguments than the ones used to prove Lemma~\ref{lem:3}, we can show that 
\begin{align}
\frac{1}{\lambda^2} \Big( -&\ln(1+\lambda(1-2A_n)) + \ln(1+\lambda(1-2B_n)) \label{eq:39} \\
&\hspace{2cm} + \frac{\text{d}}{\text{d}\alpha} \ln(1+\lambda(1-2(\alpha A_n + (1-\alpha B_n))))|_{\alpha=0} \Big) \nonumber \\
&= \int_0^{\infty} R(B_n) Q_n R(A_n) Q_n R(B_n) \text{d}t. \nonumber 
\end{align}
Now we take the trace on both sides of the above equation. On the right hand side, we interchange the trace with the integral over $t$ and use the result from Eq.~\eqref{eq:38} to arrive at
\begin{align}
\int_{-1}^1 \int_0^{\infty} \tr  &\left[ R(B) Q R(A) Q R(B) \right] \text{d}t \text{d}\mu(\lambda) \label{eq:40} \\
&= \int_{-1}^1 \frac{1}{\lambda^2} \bigg( \lim_{n \to \infty} \tr \Big[ -\ln(1+\lambda(1-2A_n)) + \ln(1+\lambda(1-2B_n)) \nonumber  \\
&\hspace{2cm} + \frac{\text{d}}{\text{d}\alpha} \ln(1+\lambda(1-2(\alpha A_n + (1-\alpha B_n))))|_{\alpha=0} \Big] \bigg) \text{d}\mu(\lambda). \nonumber
\end{align}   
Since $x \mapsto \left(-\ln(x) \right)' = -\frac{1}{x}$ is operator monotone, the integrand on the right hand side of Eq.~\eqref{eq:40} is monotone in $n$ by Theorem~\ref{thm:1}. On the other hand, from what we said above, we know that it converges pointwise for all $-1 < \lambda < 1$ as $n$ tends to infinity. Therefore, the interchange of the limit $n \to \infty$ and the integral over $\lambda$ is justified by monotone convergence. This completes the proof for the case when $Q=(A-B)$ is Hilbert-Schmidt.

Now assume that $(A-B)$ is not Hilbert-Schmidt. From \cite[Theorem~3]{Mathieu_Julien}, we conclude that there is a constant $C>0$ such that
\begin{equation}
\lim_{n \to \infty} \mathcal{H}(A_n,B_n) \geq C \left\Vert (A-B) \right\Vert_2^2 = \infty.
\label{eq:41}
\end{equation}
On the other hand 
\begin{align}
\tr \left[ R(B)Q R(A) Q R(B) \right] \geq \tr \left[ R(B) Q^2 R(B) \right] \frac{1}{4+t} = \infty,
\label{eq:42}
\end{align}
where the equality on the right hand side is justified by the fact that $R(B)$ is bounded and invertible for all $-1 < \lambda < 1$. Hence, the right hand side of Eq.~\eqref{eq:7} equals $+ \infty$ as well. This completes the proof of Theorem~\ref{thm:3}. 
\section{Proof of Theorem~\ref{thm:4}}
\label{sec:3}
As in the proof of Theorem~\ref{thm:3}, we start with a Lemma in order not to interrupt the main argumentation. Throughout the whole section we assume that the $b$ in Eq.~\eqref{eq:10} is strictly positive. This is reasonable because otherwise the relative entropy equals zero.
\begin{lemma}
Assume that $(A-B)$ and $\varphi'(B)(A-B)$ are trace-class. Then
\begin{equation}
\sum_{\beta = 1}^{\infty} \int_{-1}^{1} \left| q_{\beta} \right| \left( \psi_{\beta}, \left| \frac{2B-1}{1-\lambda (2B-1)} \right| \psi_{\beta} \right) \text{d}\mu(\lambda) < \infty,
\label{eq:43}
\end{equation}
where $(A-B) = \sum_{\beta=1}^{\infty} q_{\beta} \vert \psi_{\beta} \rangle\langle \psi_{\beta} \vert$ with $\sum_{\beta = 1}^{\infty} \left| q_{\beta} \right| < \infty$.
\label{lem:4}
\end{lemma}
\begin{proof}
The integral representation of $\varphi'$, Eq.~\eqref{eq:10}, tells us that
\begin{equation}
\varphi'(B)(A-B) = a(A-B) + b \int_{-1}^{1} \frac{2B-1}{1-\lambda (2B-1)} \text{d}\mu(\lambda) (A-B).
\label{eq:44}
\end{equation}
Because $(A-B)$ is trace-class by assumption we know that the second term on the right hand side of the above equation is trace-class as well. And due to the polar decomposition, there exist two partial isometries $U$ and $V$ such that
\begin{equation}
\int_{-1}^{1} \frac{2B-1}{1-\lambda (2B-1)} \text{d}\mu(\lambda) (A-B) = U \left| \int_{-1}^{1} \frac{2B-1}{1-\lambda (2B-1)} \text{d}\mu(\lambda) \right| \left| A-B \right| V.
\label{eq:45}
\end{equation}
Since the set of all trace-class operators is a two-sided ideal in the algebra of bounded operators $\mathcal{L}(h)$ we conclude that the term on the right hand side of Eq.~\eqref{eq:45} without $U$ and $V$ is trace-class as well. We decompose the operator $B$ in the way $B = B \ \mathds{1}(B < 1/2) + B \ \mathds{1}(B \geq 1/2)$ to see that the absolute value of the integral on the right hand side of Eq.~\eqref{eq:45} is given by $\left| \int_{-1}^{1} \frac{2B-1}{1-\lambda (2B-1)} \text{d}\mu(\lambda) \right| = \int_{-1}^{1} \left| \frac{2B-1}{1-\lambda (2B-1)} \right| \text{d}\mu(\lambda)$. Therefore,
\begin{align}
\infty > \tr \left| \int_{-1}^{1} \frac{2B-1}{1-\lambda (2B-1)} \text{d}\mu(\lambda) \right| \big| &A-B \big| \label{eq:46} \\
&= \sum_{\beta=1}^{\infty} \int_{-1}^{1} \left| q_{\beta} \right| \left( \psi_{\beta}, \left| \frac{2B-1}{1-\lambda (2B-1)} \right| \psi_{\beta} \right) \text{d}\mu(\lambda). \nonumber
\end{align}
This is what we intended to show.
\end{proof}

Having Lemma~\ref{lem:4} at hand, the proof of Theorem~\ref{thm:4}, that is the proof of the identity $\tr \left[ \varphi'(B)(A-B) \right] = \tr \left[ \frac{\text{d}}{\text{d} \alpha} \varphi \left( \alpha A + (1-\alpha)B \right) \big|_{\alpha = 0} \right]$, is in principle a straightforward computation that exploits the cyclicity of the trace. We start by inserting the integral representation of $\varphi'$ [Eq.~\eqref{eq:10}] into $\tr \left[\varphi'(B)(A-B) \right]$ to obtain  
\begin{align}
\tr \left[ \varphi'(B)(A-B) \right] &= a \tr (A-B) + b \tr \left[ \int_{-1}^1 \frac{2B-1}{1-\lambda (2B-1)} (A-B) \text{d}\mu(\lambda) \right] \label{eq:47} \\
&= a \tr (A-B) + b \sum_{\beta=1}^{\infty} \int_{-1}^1 q_{\beta} \left( \psi_{\beta}, \frac{2B-1}{1-\lambda (2B-1)} \psi_{\beta} \right) \text{d}\mu(\lambda). \nonumber
\end{align}
Here, $\left\lbrace \psi_\beta \right\rbrace_{\beta = 1}^{\infty}$ denotes the complete set of eigenfunctions of the self-adjoint operator $(A-B)$. We wish to interchange the sum over $\beta$ and the integral over $\lambda$ on the right hand side of the above equation. Using the bound 
\begin{equation}
\left| q_{\beta} \left( \psi_{\beta}, \frac{2B-1}{1-\lambda (2B-1)} \psi_{\beta} \right) \right| \leq \left| q_{\beta} \right| \left( \psi_{\beta}, \left| \frac{2B-1}{1-\lambda (2B-1)} \right| \psi_{\beta} \right) \label{eq:48}
\end{equation}
and Lemma~\ref{lem:4}, this is justified by an application of Fubini's theorem. On the other hand, the operator $\frac{2B-1}{1-\lambda (2B-1)} (A-B)$ is trace-class as long as $-1<\lambda<1$ because $(A-B)$ is trace-class and $\frac{2B-1}{1-\lambda (2B-1)}$ is bounded. We conclude that
\begin{equation}
\tr \left[ \int_{-1}^{1} \frac{2B-1}{1-\lambda (2B-1)} (A-B) \text{d}\mu(\lambda) \right] = \int_{-1}^{1} \tr \left[ \frac{2B-1}{1-\lambda (2B-1)} (A-B) \right] \text{d}\mu(\lambda). \label{eq:49}
\end{equation}
Using the identity
\begin{equation}
\frac{1-2B}{1+\lambda (1-2 B)} = \frac{1}{\lambda} - \frac{1}{\lambda} \int_{0}^{\infty} \left( \frac{1}{1+\lambda (1-2 B)+t} \right)^2 \text{d}t, 
\label{eq:50}
\end{equation}
Eq.~\eqref{eq:49} can be written as
\begin{align}
\tr \Bigg[ \int_{-1}^{1} &\frac{2B-1}{1-\lambda (2B-1)} (A-B) \text{d}\mu(\lambda) \Bigg] = \label{eq:51} \\
&= -\frac{1}{2} \int_{-1}^{1} \tr \left[ \left\lbrace \frac{2}{\lambda} - \frac{2}{\lambda} \int_{0}^{\infty} \left( \frac{1}{1+\lambda (1-2 B)+t} \right)^2 \text{d}t \right\rbrace (A-B) \right] \text{d}\mu(\lambda). \nonumber
\end{align}
With the bound
\begin{equation}
\left| \sum_{\beta=1}^{\infty} \left( \psi_{\beta}, \left( \frac{1}{1+\lambda (1-2 B)+t} \right)^2 (A-B) \psi_{\beta} \right) \right| \leq  \left( \frac{1}{1-|\lambda|+t} \right)^2 \left\Vert A-B \right\Vert_1 \label{eq:52}
\end{equation}
which holds for all $-1 < \lambda < 1$, we argue like above with Fubini that the trace can be interchanged with the integral over $t$. Now we can use the cyclicity of the trace to arrive at
\begin{align}
&\tr \left[b \int_{-1}^{1} \frac{2B-1}{1-\lambda (2B-1)} Q \text{d}\mu(\lambda) \right] = \label{eq:53} \\
&= \frac{-b}{2} \int_{-1}^{1} \left( \frac{2}{\lambda} \tr Q - \frac{2}{\lambda} \int_{0}^{\infty} \tr \left[ \frac{1}{1+\lambda (1-2 B)+t} Q \frac{1}{1+\lambda (1-2 B)+t} \right] \text{d}t \right) \text{d}\mu(\lambda). \nonumber
\end{align}
To shorten the writing, we have used the shortcut $Q=(A-B)$. Except for the fact that the trace is inside the integral, this is what we wanted to obtain (compare with the result of Lemma~\ref{lem:3}). 

Now we have to argue why we can take the trace out of the integral again which would complete the proof. By $Q_+$ and $Q_-$ we denote the positive and the negative part of the operator $Q=(A-B)$, respectively. First, we want to show that the above term with $Q$ replaced by $Q_+$ or by $Q_-$, that is
\begin{equation}
\int_{-1}^{1} \left( \frac{2}{\lambda} \tr Q_{\pm} - \frac{2}{\lambda} \int_{0}^{\infty} \tr \left[ \frac{1}{1+\lambda (1-2 B)+t} Q_{\pm} \frac{1}{1+\lambda (1-2 B)+t} \right] \text{d}t \right) \text{d}\mu(\lambda), 
\label{eq:54}
\end{equation}
is finite. To that end, we use the cyclicity of the trace to bring the two resolvents $(1+\lambda (1-2 B)+t)^{-1}$ together again, the bound from Eq.~\eqref{eq:48} [with $Q=(A-B)$ replaced by $Q_{\pm}$ on the left hand side] and Lemma~\ref{lem:4} another time. In other words, we go from Eq.~\eqref{eq:53} to Eq.~\eqref{eq:48} in backward order with $Q=(A-B)$ replaced by $Q_{\pm}$. This shows the finiteness of Eq.~\eqref{eq:54}.

Next, we go back to Eq.~\eqref{eq:54} and split the integral over $\lambda$ into three parts, one from $-1$ to $-1/2$, one from $-1/2$ to $1/2$ and a last one from $1/2$ to $1$. The integral from $-1/2$ to $1/2$ is easy to treat. We look at Eq.~\eqref{eq:54} again, adjust the boundaries of the integral over $\lambda$ to run from $-1/2$ to $1/2$ and evaluate the trace in an arbitrary basis. Like in the proof of Lemma~\ref{lem:3}, we show that there is no singularity at $\lambda = 0$. Together with the standard estimates used in the proof of Theorem~\ref{thm:3}, this implies that the expression inside the integral over $\lambda$ can be bounded by a constant. Since $\mu$ is a probability measure this is enough to apply dominated convergence and interchange the sum coming from the trace and the integral over $\lambda$. The fact that this works for any basis, shows that 
\begin{equation}
\int_{-1/2}^{1/2} \left( \frac{2}{\lambda} Q_{\pm} - \frac{2}{\lambda} \int_{0}^{\infty} \frac{1}{1+\lambda (1-2 B)+t} Q_{\pm} \frac{1}{1+\lambda (1-2 B)+t} \text{d}t \right) \text{d}\mu(\lambda), 
\label{eq:55}
\end{equation}
is trace-class and that for this term the trace and the integral can be interchanged.

In the next step, we investigate the integral from $1/2$ to $1$, that is Eq.~\eqref{eq:54} with the adjusted integral boundaries. Since 
\begin{align}
\int_{1/2}^{1} \frac{2}{\lambda} \tr Q_{\pm} \text{d}\mu(\lambda) \leq 4 \left\Vert Q \right\Vert_1
\label{eq:56} 
\end{align}
the first term inside the integral over $\lambda$ can be integrated separately. Additionally, the trace and the integral over $\lambda$ can be interchanged for this term as well. This implies that also 
\begin{equation} 
\int_{1/2}^{1} \frac{2}{\lambda} \int_{0}^{\infty} \tr \left[ \frac{1}{1+\lambda (1-2 B)+t} Q_{\pm} \frac{1}{1+\lambda (1-2 B)+t} \right] \text{d}t \text{d}\mu(\lambda)
\label{eq:57}
\end{equation}
is finite. Since the operator inside the trace is positive we can apply Fubini to interchange the trace with the integral over $t$ and afterwards with the integral over $\lambda$. The same arguments work for the integral from $-1$ to $-1/2$. Putting all this together, we have shown that
\begin{align}
b \int_{-1}^{1} \bigg( \frac{2}{\lambda} &\tr Q - \frac{2}{\lambda} \int_{0}^{\infty} \tr \left[ \frac{1}{1+\lambda (1-2 B)+t} Q \frac{1}{1+\lambda (1-2 B)+t} \right] \text{d}t \bigg) \text{d}\mu(\lambda) \label{eq:58} \\
&=\tr \left[ b \int_{-1}^{1} \left( \frac{2}{\lambda} Q - \frac{2}{\lambda} \int_{0}^{\infty} \frac{1}{1+\lambda (1-2 B)+t} Q \frac{1}{1+\lambda (1-2 B)+t} \text{d}t \right) \text{d}\mu(\lambda) \right], \nonumber
\end{align}
which together with Eq.~\eqref{eq:47} and Eq.~\eqref{eq:53} implies that
\begin{align}
\text{Tr}& \left[ \varphi'(B)(A-B) \right] \label{eq:60} \\
&=\text{Tr} \left[ a Q -\frac{b}{2} \int_{-1}^{1} \left\lbrace \frac{2Q}{\lambda} -\frac{2}{\lambda} \int_{0}^{\infty} \frac{1}{1+\lambda(1-2B) + t} Q \frac{1}{1+\lambda(1-2B) + t} \text{d}t \right\rbrace \text{d}\mu(\lambda) \right] \nonumber
\end{align}
holds. In particular, the operator on the right hand side of Eq.~\eqref{eq:60} is trace-class. Since $\varphi'(B)(A-B)$ is trace-class by assumption we know that $B$ cannot have eigenvalues at points of discontinuity of $\varphi'$ with $(A-B) \neq 0$ on the corresponding eigenspaces. From this we conclude with the help of Lemma~\ref{lem:3} that $\frac{\text{d}}{\text{d} \alpha} \varphi \left( \alpha A + (1-\alpha)B \right) |_{\alpha = 0}$ can be defined as a semibounded quadratic form on $D$. Also on $D$, it is the associated quadratic form of the operator under the trace on the right hand side of Eq.~\eqref{eq:60}, see again Lemma~3. This operator is bounded and hence we can extend $\frac{\text{d}}{\text{d} \alpha} \varphi \left( \alpha A + (1-\alpha)B \right) |_{\alpha = 0}$ to a bounded and symmetric quadratic form on all of $h$ whose associated self-adjoint operator is the operator under the trace on the right hand side of Eq.~\eqref{eq:60}. Hence, we have shown that
\begin{equation}
\tr \left[ \varphi'(B)(A-B) \right] = \tr \left[ \frac{\text{d}}{\text{d} \alpha} \varphi \left( \alpha A + (a-\alpha)B \right) \big|_{\alpha = 0} \right].
\end{equation}
This concludes the proof of Theorem~\ref{thm:4}.

\textbf{Acknowledgements.} The authors are grateful to Mathieu Lewin and Julien Sabin for useful discussions. A.D. would like to thank the IST Austria, where part of this work has been done, for its warm hospitality. Partial financial support by the DFG through the Graduiertenkolleg 1838 is gratefully acknowledged.


\newpage

(Andreas Deuchert) Mathematisches Institut, Universit{\"a}t T{\"u}bingen \\ Auf der Morgenstelle 10, 72076 T\"ubingen, Germany\\ E-mail address: \texttt{andreas.deuchert@uni-tuebingen.de} \\ \newline
(Christian Hainzl) Mathematisches Institut, Universit{\"a}t T{\"u}bingen \\ Auf der Morgenstelle 10, 72076 T\"ubingen, Germany\\ E-mail address: \texttt{christian.hainzl@uni-tuebingen.de} \\ \newline
(Robert Seiringer) Institute of Science and Technology Austria (IST Austria)\\ Am Campus 1, 3400 Klosterneuburg, Austria\\ E-mail address: \texttt{robert.seiringer@ist.ac.at}

\end{document}